\documentclass[conference]{IEEEtran}
\IEEEoverridecommandlockouts
\usepackage[T1]{fontenc}
\usepackage{cite}
\usepackage{mathtools}
\usepackage{stackengine}
\def\delequal{\mathrel{\ensurestackMath{\stackon[1pt]{=}{\scriptstyle\Delta}}}}
\usepackage{amsmath,amssymb,amsfonts}
\usepackage{amsmath,epsfig,cite,amsfonts,amssymb,psfrag,subfig}
\usepackage{graphicx}
\usepackage{textcomp}
\usepackage{xcolor}
\usepackage{algorithm}
\usepackage[noend]{algpseudocode}
\usepackage{amsthm}
\def\BibTeX{{\rm B\kern-.05em{\sc i\kern-.025em b}\kern-.08em
    T\kern-.1667em\lower.7ex\hbox{E}\kern-.125emX}}
\allowdisplaybreaks

\newtheorem{theorem}{Theorem}

\newcommand{\diag}{\mathop{\mathrm{diag}}}

\usepackage[margin=0.7in]{geometry}
\begin{document}

\title{Joint Power Allocation and Network Slicing in an Open RAN System \vspace{-.1cm}
}
%
  \author{
    \IEEEauthorblockN{Mojdeh Karbalaee Motalleb, Vahid Shah-Mansouri, Salar Nouri Naghadeh}
    \IEEEauthorblockA{School of ECE, College of Engineering, University of Tehran, Iran \\
    Email: \{mojdeh.karbalaee, vmansouri, salar.nouri\}@ut.ac.ir,
    \vspace{-.2cm}
  }
  }

\maketitle

\begin{abstract}
Open radio access network (ORAN) alliance which has been formed recently establishes a flexible,  open, and smart radio access network (RAN) by combing the ideas from xRAN and cloud RAN (C-RAN). ORAN divides the functions of the RAN into three parts, namely remote unit (RU), distributed unit (DU), and central unit (CU). While RU contains lower PHY functions, DU contains higher PHY, MAC, and RLC and CU contains RRC, PDCP, and SDAP. CU and DU are implemented as virtual network functions (VNFs) running on a cloud environment. Interface between RU, CU, and DU are open standard interfaces. Network slicing as a new concept in 5G systems is used to share the network resources between various services while the operation of one service does not affect another service. In this paper, we study the problem of RAN network slicing in an ORAN system.  We formulate the problem of wireless link scheduling, mapping the slices to the services, and mapping the physical data centers resource to slices. The objective is to jointly maximize the energy efficiency and minimize power consumption of RUs and the cost of physical resources in a downlink channel. The problem is formulated as a mixed-integer optimization problem that can be decomposed into two independent sub-problems. Heuristic algorithms are proposed for each of the sub-problems.
\end{abstract}


\section{Introduction}
Open RAN (ORAN), as the integration and expansion of C-RAN and xRAN, is expected to be a key technology in 5G networks to extensively enhance the RAN performance. The core idea of C-RAN is to split the radio remote head (RRU) from the baseband unit (BBU). Several BBUs operating on a cloud server will create a BBU-Pool, providing unified baseband signal processing with powerful computing capabilities \cite{cran1,frdl,simeone2016cloud,motalleb2017optimal}. On the other hand, xRAN technology has three fundamental features. The control plane is decoupled from user plane. Besides, a modular eNB software stack is built to operate on common-off-the-shelf (COTS) hardware. Moreover, open north-bound and south-bound interfaces are introduced \cite{xran}.

ORAN virtualizes the elements of the radio access networks, separate them and define appropriate open interfaces for connection of these elements. Moreover, ORAN uses machine learning techniques to develop smarter RAN layers in its architecture. In an innovative ORAN system, the programmable RAN software is decoupled from hardware \cite{oran1}. Open interface is one of the most crucial properties for ORAN to enable mobile network operators (MNOs) to define their own services.
The concept of software defined network (SDN), which is the separation of control plane from user plane, is deployed in an intelligent ORAN architecture. Moreover, this separation promotes RRM to use non-realtime (RT) and near-realtime (NRT) RAN Intelligent Controller (RIC).
In the ORAN architecture, distributed unit (DU) is a logical node having RLC/MAC/High-PHY layer.
Moreover, central unit (CU) is a logical node having RRC, SDAP and PDCP. Also, radio unit (RU) is a logical node having LOW-PHY layer and
RF processing \cite{oranpaper}.
ORAN introduces interfaces such as open fronthaul interface that connects DU and RU (i.e., E2 interface), and an A1 interface between orchestration/NMS layer containing the non-real-time RIC (RIC non-RT) function and the eNB/gNB containing the near-real-time RIC (RIC near-RT) function.

To evolve servicing in 5G networks, separation of elements of software and hardware of network is employed and referred to as network functions virtualization (NFV). Virtual network function (VNF) are functional blocks of the system. 5G networks are expected to host several services with different requirements simultaneously. Network slicing is considered as a solution for such demand. A network slice is a logical end-to-end network which provides service with specific requirements. Multiple network slices run on the same network infrastructure which are managed and operate independently \cite{lee2018dynamic}. Slicing can be defined for the RAN, for the core network or both of these them \cite{ns1}.

Based on their type of request, UEs are classified into various services, where UEs in the same service have similar service requests. In addition, each service is mapped to one or more slices based on the resources of the slices. Each VNF of the ORAN is mapped to one or some virtual machines (VMs) in the data center. For simplicity of analysis and without loss of generality, we assume a VNF is mapped to one VM and requires specific physical resources including storage, memory, and processing \cite{frdl,luong2018novel,luong2018novel1}.

In \cite{oranT},  the evolution of RAN in the  5G network is described. The architecture of C-RAN and X-RAN is expressed and development of the ORAN from C-RAN and X-RAN is explained. Also requirements of the network in terms of capacity and latency is taken into account.
In \cite{oran12},  the ORAN architecture is studied. The eight work groups is discussed including: focus areas of ORAN alliance which contains use-cases, overall architecture,  open interfaces, and  the virtualization and modularization of hardware and software. To the best of our knowledge, there is not work in the literature to model ORAN systems. 


\begin{figure}
  \centering
    \includegraphics[scale=0.55]{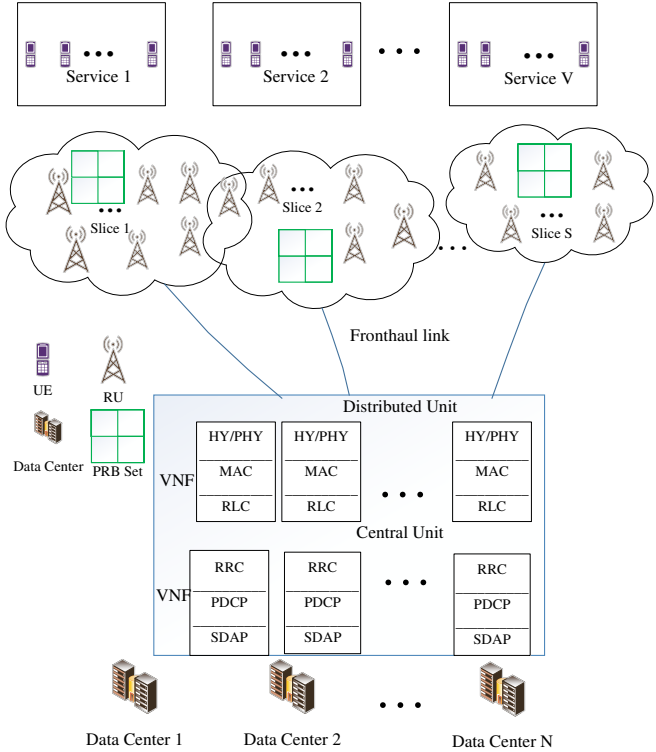}
  \caption{Network sliced ORAN system}
  \label{fig:c11}
\end{figure}

In this paper, as  depicted in Figure \ref{fig:c11}, the downlink of the ORAN system is studied. The RAN is divided into three layers including RU, DU, and CU with open interfaces. DU and CU are running on general purpose data centers. UEs are divided to different services according to their requirements. RAN resources are decoupled to slices to provide requirements of services. Optimal power allocation and joint mapping of slices to services are applied. In addition, mapping slices to physical resources is taken to account. The contributions of the paper are as follows:
\begin{itemize}
\item In this paper, joint network slicing and resource allocation is considered in an ORAN system.
\item We formulate the problem of assigning UEs to services, services to slices, and physical wireless and data center resources to the slices as an optimization problem.
\item The problem is decomposed into two independent sub-problems.
\item Novel heuristic algorithms are applied for these sub-problems to efficiently obtain the solution.
\end{itemize}

The rest of the paper is organized as follows. In Section \ref{systemmodel}, system model which contains obtaining achievable rates, processing and transmission delay, and physical data center resources is introduced. Then, problem statement is explained and decomposed into two independent sub-problems.
In Section \ref{proposedmethod}, heuristic algorithms for the sub-problems is presented. In Section \ref{simul}, numerical results are provided to investigate the performance of the algorithms.

\section{System Model and Problem Formulation}\label{systemmodel}

In this section, first, we  present the system model. Then, we obtain achievable rates and delays for the downlink (DL) of the ORAN system. Afterward, we discuss about assignment of physical data center resources.
Finally, the main problem is expressed.
\subsection{System Model}
Suppose there are $S$ slices Serving $V$ services. Each Service $v\in \{1,2,...,V \} $ consists of $U_v$
single-antenna user equipments (UEs) that require certain service. Each slice $s \in \{1,2,...,S \}$ consists of $R_s$ RUs and $K_s$ physical resource blocks (PRBs), one DU and one CU that contains VNFs.
Slices can have shared resources. All RUs in a slice, that is mapped to a service, transmit signals cooperatively to all the UEs in a specific service \cite{motalleb2017optimal,mimoCran}. Each RU $r \in \{1,2,...,R \}$ is mapped to a DU via an optical fiber link with limited fronthaul capacity.
There are two processing layers one in the DU and one in the CU of ORAN system, each represented with a VNF. The lower layer (i.e., DU) consists of high-PHY, MAC, and RLC, and the upper layer (i.e., CU) consists of RRC, PDCP and SDAP. Assume we have $M_1$ VNFs in the DU layer and $M_2$ VNFs in the CU layer for processing data.
Each VNF in both layers belongs to one or more slices. So, in the $s^{th}$ slice, there are $M_{s,1}$ VNFs in the DU layer and $M_{s,2}$ VNFs in the CU layer. The VNFs in the DU and CU layers have the computational capacity that is  equal to $\mu_1$ and $\mu_2$, respectively.
Also, RUs and PRBs can serve more than one slices.
\subsection{The Achievable Rate}
The achievable data rate for the $i^{th}$ UE in the $v^{th}$ service can be written as
\begin{equation}\label{eq1}
\mathcal{R}_{u(v,i)} = B \log_2({1+ \rho_{u(v,i)}}),
\end{equation}
where $B$ is the bandwidth of system and $\rho_{u(v,i)}$ is the SNR of $i^{th}$ UE in $v^{th}$ service which is obtained from
\begin{equation}\label{eq2}
\rho_{u(v,i)} =  \frac{p_{u(v,i)}\sum_{s=1}^{S}|\bold{h}_{R_s,u(v,i)}^H \bold{w}_{R_s,u(v,i)}|^2 a_{v,s}}{BN_0 + I_{u(v,i)}},
\end{equation}
where $p_{u(v,i)}$ represents the transmission power allocated by RUs to $i^{th}$ UE in $v^{th}$ service, and
$\bold{h}_{R_s,u(v,i)} \in \mathbb{C}^{{R}_s}$ is the vector of channel gain of a wireless link from RUs in the $s^{th}$ slice to the $i^{th}$ UE in $v^{th}$ service. In addition, $\bold{w}_{R_s,u(v,i)} \in \mathbb{C}^{{R}_s}$ depicts the  transmit beamforming vector from RUs in the $s^{th}$ slice to the $i^{th}$ UE in $v^{th}$ service. Moreover, $BN_0$ denotes the power of Gaussian additive noise, and $I_{u(v,i)}$ is the power of interfering signals. Moreover, $a_{v,s} \in \{0,1\}$ is a binary variable that illustrates whether slice $s$ is mapped to service $v$ or not. If $a_{v,s} =1$ then, $v^{th}$ service is mapped to $s^{th}$ slice; otherwise, it is not mapped.
\newline
To obtain SNR as formulated in \eqref{eq2}, let $\bold{y}_{U_v}\in \mathbb{C}^{U_v} $ be the received signal's vector of all users in $v^{th}$ service
\begin{equation}\label{eq3}
\textstyle \bold{y}_{U_v} = \sum_{s = 1}^{S}\sum_{k=1}^{K_s} \boldsymbol{H}^H_{\mathcal{R}_s,\mathcal{U}_v} \
\mathfrak{y}_{R_s}\zeta_{U_v,k,s} a_{v,s}+ \boldsymbol{z}_{\mathcal{U}_v},
\end{equation}
where $\mathfrak{y}_{R_s} = \boldsymbol{W}_{\mathcal{R}_s,\mathcal{U}_v}\boldsymbol{P}_{U_v}^{\frac{1}{2}}\boldsymbol{x}_{\mathcal{U}_v}+ \boldsymbol{q}_{\mathcal{R}_s}$
and $\boldsymbol{x}_{ \mathcal{U}_v} = [x_{ u_{(v,1)}},...,x_{ u_{(v,\mathcal{U}_v)}}]^T \in \mathbb{C}^{{R}_s } $ depicts the transmitted symbol vector of UEs in $v^{th}$ set of service,  $\boldsymbol{z}_{U_v}$ is the additive Gaussian noise $\boldsymbol{z_{U_v}} \backsim \mathcal{N}(0,N_0\boldsymbol{I}_{{U}_v})$ and $N_0$ is the noise power.
In addition, $\boldsymbol{q}_{R_s} \in \mathbb{C}^{{R}_s }  $ indicates the quantization noise, which is made from signal compression in DU.
Besides, $\boldsymbol{P}_{U_v} = \diag{(p_{u_{(v,1)}}, ..., p_{u_{(v,\mathcal{U}_v)}})}$.
\newline
Furthermore, $\zeta_{k,s}^{U_v} \delequal \{\zeta_{k,s}^{u(v,1)},\zeta_{k,s}^{u(v,2)},...,\zeta_{k,s}^{u(v,N_{U_v})}\}$,
$\zeta_{k,s}^{u(v,i)} \in \{0,1\}$ is a binary parameter, which demonstrates whether $i^{th}$ UE in $v^{th}$ service can transmit its signals through $k^{th}$ PRB and also this PRB belongs to $s^{th}$ slice or not.
$\boldsymbol{H}_{\mathcal{R}_s,\mathcal{U}_v}=\left[\boldsymbol{h}_{\mathcal{R}_s,u_{(v,1)}},\ldots,\boldsymbol{h}_{\mathcal{R}_s,v_{(v,\mathcal{U}_v)}}\right]^T  \in \mathbb{C}^{{R}_s\times {U}_v }$
shows the channel matrix between RU set $\mathcal{R}_s$ to UE set
$\mathcal{U}_v$, besides.
What's more, it is assumed we have perfect channel state information (CSI).\newline
Moreover, $\boldsymbol{W}_{\mathcal{R}_s,\mathcal{U}_v} = [\boldsymbol{w}_{\mathcal{R}_s,u(v,1)},...,\boldsymbol{w}_{\mathcal{R}_s,u(v,U_v)}] \in \mathbb{C}^{{R}_s\times U_v} $ is the zero forcing beamforming vector to minimize the interference which is indicated as below
\begin{equation}
\textstyle \boldsymbol{W}_{\mathcal{R}_s,\mathcal{U}_v} = \boldsymbol{H}_{\mathcal{R}_s,\mathcal{U}_v}(\boldsymbol{H}_{\mathcal{R}_s,\mathcal{U}_v}^H \boldsymbol{H}_{\mathcal{R}_s,\mathcal{U}_v})^{-1}.
\end{equation}
Hence, the interference power of $i^{th}$ UE in $v^{th}$ service can be represented as follow
\begin{equation}
\begin{split}
 I_{u_{(v,i)}} &=
 \underbrace{\sum_{s=1}^{S}\sum_{n=1}^{S}\sum_{\substack{l=1 \\ l\neq i}}^{{U}_v} \gamma_{1}  p_{u_{(v,l)}}a_{v,s}\zeta_{u_(v,i),n,s}\zeta_{u_(v,l),n,s}}_{\text{(intra-service interference)}}\\
&+ \underbrace{\sum_{\substack{y=1 \\ l\neq v}}^{V}\sum_{s=1}^{S}\sum_{n=1}^{S}\sum_{l=1}^{{U}_y} \gamma_{2}  p_{u_{(y,l)}}a_{y,s} \zeta_{u_(v,i),n,s}\zeta_{u_(y,l),n,s}}_{\text{(inter-service interference)}}\\
&+\underbrace{ \sum_{s=1}^{S} \sum_{j=1}^{{R}_s} {\sigma_q}_{r_{(s,j)}}^2 |\boldsymbol{h}_{r_{(s,j)}, u_{(v,i)}}|^2 a_{v,s}}_{\text{(quantization noise interference)}},
\end{split}
\end{equation}
where $\gamma_{1} =|\boldsymbol{h}_{\mathcal{R}_s, u_{(v,i)}}^H \boldsymbol{w}_{\mathcal{R}_{s},u_{(v,l)}}|^2$
and $\gamma_{2} =|\boldsymbol{h}_{\mathcal{R}_s, u_{(v,i)}}^H \boldsymbol{w}_{\mathcal{R}_{s},u_{(y,l)}}|^2$. Moreover,
${\sigma_q}_{r_{(s,j)}}$ is the variance of quantization noise of $j^{th}$ RU in $s^{th}$ slice.
Interference signal for each UE is coming from UEs using the same PRB.
If we replace $p_{u_{(v,l)}}$ and $p_{u_{(y,l)}}$ by $P_{max}$, an upper bound $\bar{I}_{u_{(v,i)}}$ is obtained for $I_{u_{(v,i)}}$. Therefore, $\bar{\mathcal{R}}_{u_{(v,i)}} \forall v , \forall i$ is derived by using $\bar{I}_{u_{(v,i)}}$ instead of $I_{u_{(v,i)}}$ in  \eqref{eq1} and \eqref{eq2}.\newline
Let $\bar{p}_{r_{(s,j)}}$ denote the power of transmitted signal from the $j^{th}$ RU in $s^{th}$ slice.
From \eqref{eq3}, we have,
\begin{equation}
\bar{p}_{r_{(s,j)}} = \sum_{v=1}^{V}\boldsymbol{w}_{r_{(s,j)},\mathcal{U}_{v}} \boldsymbol{P}_{\mathcal{U}_v}^{\frac{1}{2}} \boldsymbol{P}_{\mathcal{U}_v}^{H \frac{1}{2}}   \boldsymbol{w}_{r_{(s,j)},\mathcal{U}_{v}}^H a_{v,s} + \sigma_{q_{r(s,j)}}^2.
\end{equation}
Nevertheless, the rate of users on the fronthual link between DU and the $j^{th}$ RU in $s^{th}$ slice is formulated as \cite{simeone2016cloud, 1111}
\begin{equation}
C_{R_{(s,j)}} = \log{(1+\sum_{v=1}^{V}\frac{w_{r_{(s,j)},\mathcal{D}_{s}} \boldsymbol{P}_{\mathcal{U}_v}^{\frac{1}{2}} \boldsymbol{P}_{\mathcal{U}_v}^{H \frac{1}{2}}   w_{r_{(s,j)},\mathcal{U}_{v}}^H a_{v,s}}{ \sigma_{q_{r(s,j)}}^2})},
\end{equation}
where, $a_{v,s}$ is a binary variable denotes whether the slice $s$ is mapped to service $v$ or not .

\subsection{Mean Delay}
Assume the packet arrival of UEs follows  a Poisson process with arrival rate $\lambda_{u(v,i)}$ for the $i^{th}$ UE of the $v^{th}$ service.
Therefore, the mean arrival data rate of UEs mapped to the $s^{th}$ slice in the CU layer is $\alpha_{s_1} = \sum_{v=1}^{V}\sum_{u=2}^{U_v}a_{v,s}\lambda_{u(v,i)}$, where $a_{v,s}$ is a binary variable which indicates whether the $v^{th}$ service is mapped to the $s^{th}$ slice or not.
Furthermore, the mean arrival data rate of the DU layer is approximately equal to the mean arrival data rate of the first layer $\alpha_{s} =\alpha_{s_1} \approx \alpha_{s_2}$ since, by using Burke’s Theorem, the mean arrival data rate of the second layer which is processed in the first layer is still Poisson with rate $\alpha_{s}$.
It is assumed that there are load balancers in each layer for each slice to divide the incoming traffic to VNFs equally \cite{frdl,luong2018novel,luong2018novel1}.
Suppose the baseband processing of each VNF is depicted as an M/M/1 processing queue.
Each packet is processed by one of the VNFs of a slice. So, the mean delay of the $s^{th}$ slice in the first and the second layer, modeled as M/M/1 queue, is formulated as follow, respectively
\begin{equation}
\begin{split}
d_{s_1} &= \frac{1}{\mu_1 - \alpha_{s}/{M_{s,1}}},\\
d_{s_2} &= \frac{1}{\mu_2 - \alpha_{s}/{M_{s,2}}}.
\end{split}
\end{equation}
where $1/\mu_1$ and $1/\mu_2$ are the mean service time of the first and the second layers respectively.
Besides, $\alpha_{s}$ is the  arrival rate which is divided
by load balancer before arriving to the VNFs. The  arrival rate of each VNF in each layer of the slice $s$ is $\alpha_{s}/{M_{s,i}}$ $ i \in \{1,2\}$.
In addition, $d_{s_{tr}}$ is the transmission delay for $s^{th}$ slice on the  wireless link. The arrival data rate of wireless link
 is equal to the arrival data rate of load balancers for each slice \cite{frdl}.
Moreover, it is assumed that the service time of transmission queue for each slice $s$ has
 an exponential distribution with mean $1/(R_{{tot}_s})$ and can be modeled as a M/M/1 queue \cite{frdl,luong2018novel,luong2018novel1,guo2016exploiting}. Therefore,
the mean delay of the transmission layer is
\begin{equation}
 d_{s_{tr}} = \frac{1}{R_{{tot}_s} - \alpha_{s}};
\end{equation}
where, $R_{{tot}_s} =  \sum_{v=1}^{V}\sum_{u=2}^{U_v}a_{v,s}R_{u(v,i)}$ is the total achievable rate of each slice that is mapped to specific service.
Mean delay of each slice is
\begin{equation}
D_{s} = d_{s_1} + d_{s_2} + d_{s_{tr}} \forall s.
\end{equation}
\subsection{Physical Data Center Resource}
Each VNF requires
physical resources that contain memory, storage and CPU.
Let the required resources for VNF $f$ in slice $s$ is represented by a tuple as
\begin{equation}
\bar{\Omega}_{s}^f = \{\Omega_{M,{s}}^f, \Omega_{S,{s}}^f, \Omega_{C,{s}}^f \},
\end{equation}
where $\bar{\Omega}_{s}^f\in \mathbb{C}^{3}$ and $\Omega_{M,{s}}^f, \Omega_{S,{s}}^f, \Omega_{C,{s}}^f$ indicate the amount of required memory, storage, and CPU, respectively.
Moreover, the total amount of required memory, storage and CPU of all VNFs of a slice is defined as
\begin{equation}
\textstyle \bar{\Omega}_{\mathfrak{z},s}^{tot} = \sum_{f=1}^{M_{s_1} + M_{s_2}}\bar{\Omega}_{\mathfrak{z},s}^f \;\; \mathfrak{z} \in \{M, S, C\}.
\end{equation}
Also, there are $D_c$ data centers (DC), serving the VNFs. Each DC contains several servers that supply VNF requirements.
The amount of memory, storage and CPU is denoted by $\tau_{M_{j}}, \tau_{S_{j}}$and $\tau_{C_{j}} $ for the $j^{th}$ DC, respectively
\begin{equation*}
\tau_j = \{\tau_{M_{j}}, \tau_{S_{j}}, \tau_{C_{j}} \},
\end{equation*}
In this system model, the assignment of physical DC resources to VNFs is considered. Let $y_{s,d}$ be a binary variable indicating whether the $d^{th}$ DC is connected to the VNFs of $s^{th}$ slice or not.
\subsection{Problem Statement}
An important criterion to measure the optimality of a system is energy efficiency represented as the sum-rate to sum-power
\begin{equation}
\textstyle \eta(\boldsymbol{P},\boldsymbol{A}) := \frac{\sum\limits_{v=1}^{V} \sum\limits_{k=1}^{{U}_v}\mathcal{R}_{u_{(v,k)}} }{\sum\limits_{s=1}^{S} \sum\limits_{i=1}^{{R}_s}\bar{p}_{r_{(s,i)}}} = \frac{\mathfrak{R}_{tot}(\boldsymbol{P},\boldsymbol{A})}{P_r^{{tot}}(\boldsymbol{P},\boldsymbol{A})},
\end{equation}
where, $P_r^{tot}(\boldsymbol{P},\boldsymbol{A}) = \sum\limits_{s=1}^{S}\sum\limits_{i=1}^{{R}_s}\bar{p}_{r_{(s,i)}}$ is the total power consumption of all RUs in all slices. Also, $\mathfrak{R}_{tot}(\boldsymbol{P},\boldsymbol{A}) = \sum\limits_{v=1}^{V} \sum\limits_{k=1}^{{U}_v}\mathcal{R}_{u_{(v,k)}} $ is the toal rates of all UEs applied for all types of services.
Assume the power consumption of baseband processing at each DC $d$ that is connected to VNFs of a slice $s$ is depicted as
$\phi_{s,d}$. So the total power of the system for all active DCs that are connected to slices can be represented as
\begin{equation*}
\textstyle \phi_{tot} = \sum_{s=1}^{S}\sum_{d=1}^{D_c}y_{s,d}\phi_{s,d}.
\end{equation*}
Also, a cost function for the placement of VNFs into DCs is defined as
\begin{equation}\label{eqpsi}
\textstyle  \psi_{tot} = \phi_{tot} - \nu \sum_{d=1}^{D_c}\sum_{v=1}^{V}y_{s,d}a_{v,s}
\end{equation}
where, $\nu$ is a design variable to value between the first term of \eqref{eqpsi} which is the total power consumption of physical resources and the second term that is shown the amount of admitted slices to have physical resources.
Our goal is to maximize sum-rate and minimize sum-power (the total power of all RUs and the total power consumption of baseband processing at all DCs) simultaneously, with the presence of constraints which is written as follow,
\begin{subequations}
\begin{alignat}{4}
\max\limits_{\boldsymbol{P}, \boldsymbol{A}, \boldsymbol{Y} }   \quad &   \eta(\boldsymbol{P},\boldsymbol{A})+\frac{1}{\psi_{tot}(\boldsymbol{Y})} \\
\text{subject to} \quad  & \bar{p}_{r_{(s,i)}} \leq P_{max} \quad \forall s, \forall i,
 \label{c11} \\
&p_{u_{(v,k)}}  \geq 0  \quad \forall v, \forall k,\label{c12} \\
&\mathcal{R}_{u_{(v,k)}} \geq  \mathcal{R}_{u_{(v,k)}}^{min} \quad \forall v, \forall k,\label{c13} \\
&C_{r_{(s,i)}} \leq C_{r_{(s,i)}}^{max} \quad \forall s, \forall i, \label{c14}\\
&D_{s} \leq D_{s}^{max} \quad \forall s,\label{c15} \\
& \textstyle  \sum_{s=1}^{S}a_{v,s} \geq 1 \quad \forall s, \label{c21} \\
& \textstyle  \sum_{d=1}^{D_c}\sum_{v=1}^{V}y_{s,d}a_{v,s} \geq 1\times\sum_{v=1}^{V}a_{v,s} \forall s,\label{c23} \\
 &\textstyle  \bar{\Omega}_{\mathfrak{z},s}^{tot} = \sum_{f=1}^{F_s}\bar{\Omega}_{\mathfrak{z},s}^f \leq  \sum_{d=1}^{D_c} y_{s,d} \tau_{\mathfrak{z}_d}  \forall s, \forall \mathfrak{z}\in \mathcal{E}; \label{c22}
\end{alignat}
\label{constraints}
\end{subequations}
where $\boldsymbol{P} =[p_{u(v,k)}] \:\: \forall v , \forall k $, is the matrix of power for UEs, $\boldsymbol{A} =[a_{v,s}] \:\: \forall v , \forall s $ denotes the binary variable for connecting slices to services and $\boldsymbol{Y} =[y_{s,d}]  \:\: \forall s ,  \forall d $ is a binary variable shown whether
the physical DC is mapped to a VNFs of a slice or not.
\eqref{c11}, and \eqref{c12}, indicate that the power of each RU do not exceed the maximum power, and the power of each UE is a positive integer value, respectively. Also \eqref{c13} shows that the rate of each UE is more than a threshold. \eqref{c14} and \eqref{c15} expressed the limited capacity of the fronthaul link, and the limited delay of receiving signal, respectively.
Furthermore, \eqref{c21}
ensures that each service is mapped to at least one slice.
Also, \eqref{c23}, guarantees that each slice (VNFs in two layers of slices) has been placed to one or more physical resources of DCs. Moreover, in \eqref{c22}  $\mathcal{E} = \{M,S,C\}$ and the constraint supports
that we have enough physical resources for VNFs of each slice.\newline
The optimization problem in \eqref{constraints} can be decomposed into two independent optimization problems A and B since the variables can be obtained independently and respectively. Firstly we need to solve problem A. After obtaining $\boldsymbol{P}$ and $ \boldsymbol{A}$, problem B can be solved by having the value of $ \boldsymbol{A}$. 
The problem A is as follow
\begin{subequations}
\begin{alignat}{4}
\max\limits_{\boldsymbol{P}, \boldsymbol{A} }   \quad &   \eta(\boldsymbol{P},\boldsymbol{A})\\
\text{subject to} \quad  & \bar{p}_{r_{(s,i)}} \leq P_{max} && \quad \forall s, \forall i,   \\
&p_{u_{(v,k)}}  \geq 0  &&\quad \forall v, \forall k, \\
&\mathcal{R}_{u_{(v,k)}} \geq  \mathcal{R}_{u_{(v,k)}}^{min} && \quad \forall v, \forall k, \\
&C_{r_{(s,i)}} \leq C_{r_{(s,i)}}^{max}  &&\quad \forall s, \forall i,\label{cc14} \\
&D_{s} \leq D_{s}^{max}  &&\quad \forall s, \label{cc15} \\
& \textstyle  \sum_{s=1}^{S}a_{v,s} \geq 1 &&\quad \forall s.
\end{alignat}
\label{constraints1}
\end{subequations}
In problem B, $ \boldsymbol{Y}$ is obtained. The problem B is
\begin{subequations}
\begin{alignat}{4}
\min\limits_{\boldsymbol{y} }   \quad &   \psi_{tot}(\boldsymbol{Y})\\
\text{s. t.} \quad & \textstyle \sum_{d=1}^{D_c}\sum_{v=1}^{V}y_{s,d}a_{v,s} \geq 1\times\sum_{v=1}^{V}a_{v,s} \forall s, \\
 &\textstyle  \bar{\Omega}_{\mathfrak{z},s}^{tot} = \sum_{f=1}^{F_s}\bar{\Omega}_{\mathfrak{z},s}^f \leq  \sum_{d=1}^{D_c} y_{s,d} \tau_{\mathfrak{z}_d}
\forall s, \forall \mathfrak{z}\in \mathcal{E};  \label{eqomega}
\end{alignat}
\label{constraints2}
\end{subequations}
\section{Heuristic Method}\label{proposedmethod}
In this subsection, a heuristic  method is proposed to solve the problem in \eqref{constraints1} which is non-convex and computationally hard. To make it tractable, we divide problem \eqref{constraints1} into
two different part that can be solved iteratively.
In the first part of sub-problem A, we obtain $\boldsymbol{A}$ by fixing $\boldsymbol{P} = P_{max}$ in \eqref{constraints1}. Also, we set $\eta = 0$. Afterward, by achieving $\boldsymbol{A}$, in the second part, we find $\boldsymbol{P}$ by using \eqref{constraints1}. This part of the problem can be approximated and converted to a convex problem, so the problem can be solved by convex methods. After solving $\boldsymbol{P}$, $\eta$ is updated. Then in the next iteration, with new $\boldsymbol{P}$
and $\eta$, the two parts of the problems are solved. We repeat this procedure until the algorithm  converges.
\subsection{First Part of Sub-Problem A}\label{firstsub}
Two different methods are applied to acquire $\boldsymbol{A}$.
The details of the heuristic algorithm are represented in algorithm \ref{alg}.
\begin{algorithm}
\caption{Mapping Slice to Service}\label{alg}
\begin{algorithmic}[1]
\State Sort services according to the number of UEs in it and their requirements in the descending order.
\State Sort slices according to the weighted linear combination of number of PRBs, RUs and VNFs in two layers and the capacity of their resources in the descending order.
\For {$i \gets 1$ to $S$}
\For {$j \gets 1$ to $V$}
\State Set $a_{i,j} = 1$
\State Obtain Parameters of Systems (power and rate of UEs, rate of fronhaul links, power of RUs)
\If {conditions \eqref{c11}, \eqref{c12}, \eqref{c13} and \eqref{c14} is not applied}
\State Set $a_{i,j} = 0$;
\Else
\State break from inner loop;
\EndIf
\State \textbf{end if}
\EndFor
\State \textbf{end for}
\EndFor
\State \textbf{end for}
\end{algorithmic}
\end{algorithm}
\subsection{Second Part of Sub-Problem A}\label{secondsub}
In this part, by assuming that $\boldsymbol{A}$ is fixed, the power of UEs in each service is achieved.
\begin{theorem}\label{t2}
The optimum energy efficiency is achieved if
\begin{equation}\label{q2}
\max \limits_{\boldsymbol{P}} (\mathfrak{R}_{tot}(\boldsymbol{P}) - \eta^* P_{r_{tot}}(\boldsymbol{P}))=
 \mathfrak{R}_{tot}(\boldsymbol{P}^*) - \eta^* P_{r_{tot}}(\boldsymbol{P}^*) =0.
\end{equation}
\end{theorem}
\begin{proof}
See \cite[Appendix A]{aaa}
\end{proof}
The second sub-problem can  be solved using the Lagrangian function and iterative algorithm.
Since, Interference is a function of the power of UEs, to make it tractable, we assume an upper bound $\bar{I}_{u_{(v,i)}}$ for interference. In order to make \eqref{constraints1} as a standard form of a convex optimization problem, it is required to change the variable of equations \eqref{cc14} and \eqref{cc15} ($p_{r_{(s,i)}} = \sigma_{q_{r(s,j)}}^2\times 2^{C_{r_{(s,i)}}}$ and $1/(D_{s}- d_{s_1} + d_{s_2})+\alpha_s$ respectively).
Assume $\boldsymbol{\lambda}$, $\boldsymbol{\mu}$, $\boldsymbol{\xi}$, and $\boldsymbol{ \kappa}$ are the matrix of Lagrangian multipliers that have non-zero positive elements.
The Lagrangian function is written as follow
\begin{subequations}\label{lagrang}
\begin{alignat}{4}
\mathcal{L}(\boldsymbol{P}; \boldsymbol{\lambda}, \boldsymbol{\chi}, \boldsymbol{\mu}, \boldsymbol{ \xi}, \boldsymbol{ \kappa}) & = \sum\limits_{v=1}^{V} \sum\limits_{k=1}^{U_v}\mathcal{\bar{R}}_{u_{(v,k)}}
- \eta \sum\limits_{v=1}^{V} \sum\limits_{i=1}^{\mathcal{R}_s}\bar{p}_{r_{(s,i)}}\\
&+\sum\limits_{v=1}^{V} \sum\limits_{k=1}^{U_v} \lambda_{u_{(v,k)}} (\mathcal{\bar{R}}_{u_{(v,k)}}-\mathcal{R}_{u_{(v,k)}}^{max})\\
&- \sum\limits_{s=1}^{S} \sum\limits_{i=1}^{R_s} \mu_{r_{(s,i)}} (\bar{p}_{r_{(s,i)}}-P_{max})\\
&- \sum\limits_{s=1}^{S} \sum\limits_{i=1}^{R_s} \xi_{r_{(s,i)}} (\bar{p}_{r_{(s,i)}}-\sigma_{q_{r(s,i)}}^2 2^{C_{r_{(s,i)}}^{max}}).\\
&+ \sum\limits_{v=1}^{V} \sum\limits_{k=1}^{U_v} \kappa_{u_{(v,k)}} \sum\limits_{s=1}^{S}(R_{u_{(v,k)}} -\mathfrak{D_s})a_{v,s}.
\end{alignat}
\end{subequations}
where, $\mathfrak{D_s}=\frac{1}{D_{s}^{max}-d_{s_1}-d_{s_2}}+\alpha_s$. Optimal power is obtained from  \eqref{lagrang}
\begin{equation}
p_{u(v,i)}^{*} = [\frac{\mathfrak{y}_{u(v,i)}\mathfrak{w}_{u(v,i)}-\mathfrak{x}_{u(v,i)}\mathfrak{z}_{u(v,i)}}{\mathfrak{x}_{u(v,i)}\mathfrak{w}_{u(v,i)} }]^+
\end{equation}
where, $\mathfrak{y}_{u(v,i)}= (\lambda_{u(v,i)}+\kappa_{u_{(v,k)}}+1)\frac{B}{Ln_2}$ and
$\mathfrak{w}_{u(v,i)} = \sum_{s=1}^{S}|\bold{h}_{R_s,u(v,i)}^H \bold{w}_{R_s,u(v,i)}|^2 a_{v,s}$. Also
$\mathfrak{z}_{u(v,i)} = BN_0 + \bar{I}_{u(v,i)}$ and $\mathfrak{x}_{u(v,i)} = \sum\limits_{s=1}^{S} \sum\limits_{i=1}^{R_s} ( \mu_{r_{u(s,i)}} + \xi_{r_{(s,i)}}+\eta)||w_{r_{(s,j)},u_{(v,i)}}||^2$.
By using sub-gradient method, the optimal power $\boldsymbol{P}$ is obtained \cite{mimoCran}.
\subsection{Solving two part of Sub-problem A iteratively}
In \eqref{firstsub} and \eqref{secondsub}, the details of solving each part of the sub-problem are depicted.
Firstly, we obtain $\boldsymbol{A}$ by fixing $\boldsymbol{P} = P_{max}$ in the problem \eqref{constraints1} and using algorithm \eqref{alg}. Also we fixed $\eta = 0$. Afterward, $\boldsymbol{A}$ is achieved from algorithm \ref{alg}. Then, in the second part, we acquire $\boldsymbol{P}$ using the sub-gradient method. After solving $\boldsymbol{P}$, $\eta$ is updated. Then in the next iteration, with new $\boldsymbol{P}$
and $\eta$ the two parts of the problems is solved until the algorithm converges.
Here, the algorithm of solving sub-problem A is shown in algorithm \eqref{alg2}
\begin{algorithm}
\caption{Joint Network Slicing and Power Allocation}\label{alg2}
\begin{algorithmic}[1]
\State Set the maximum number of iterations $I_{max}$, convergence condition $\epsilon_{\eta}$  and the initial value $\eta^{(1)} = 0$
\State Set $\boldsymbol{P} = \boldsymbol{P}_{max}$
\For {$counter \gets 1$ to $I_{max}$}
\State Achieve $\boldsymbol{A}$ by applying Algorithm \eqref{alg}
\State Obtain $\boldsymbol{P}$ by using sub-gradient method which is mentioned in \eqref{secondsub}.
\If {$ \mathfrak{R}_{tot}(\boldsymbol{P}^{(i)},\boldsymbol{A}^{(i)}) - \eta^{(i)} P_{r_{tot}}(\boldsymbol{P}^{(i)},\boldsymbol{A}^{(i)}) < \epsilon_{\eta} $}
\State Set $\boldsymbol{P}^*= \boldsymbol{P}^{(i)} $, $\boldsymbol{A}^*= \boldsymbol{A}^{(i)} $   and  $ \eta^{*} =\eta^{(i)} $;
\State break;
\Else
\State $i= i+1$, Setting $\boldsymbol{P} = \boldsymbol{P}^{(i)}$ ;
\EndIf
\State \textbf{end if}
\EndFor
\State \textbf{end for}
\end{algorithmic}
\end{algorithm}
\subsection{Sub-Problem B}
In this subsection, we would like to solve  \eqref{constraints2}, which is the placement of virtual resources to physical resources in order to minimize the cost function $\psi_{tot}$.
To achieve optimum $\boldsymbol{Y}$ heuristic algorithm is applied. The details of the heuristic algorithm are written in algorithm \eqref{alg3}. In this algorithm, firstly, we sort slices and DCs according to their sum-weighted of their requirements (line \ref{31} and line \ref{32} of algorithm \ref{alg3}).
We define a weighted parameter for $\Omega_{\mathfrak{z},s}^{tot}$ and $\tau_j^\mathfrak{z}$ as follow
\begin{equation}\label{wt}
\begin{split}
\hat{\Omega}_{s}^{tot} &= w_M \bar{\Omega}_{M,s}^{tot} + w_S \bar{\Omega}_{S,s}^{tot} + w_C \bar{\Omega}_{C,s}^{tot} \\
\hat{\tau}_j &= w_M \tau_{{j}}^M + w_S \tau_{{j}}^S + w_C \tau_{{j}}^C,
\end{split}
\end{equation}
where, $\boldsymbol{w} = \{w_M, w_S, w_C\}$ is the weight of memory, storage and CPU.
Secondly, we start mapping from the most needed slices to the DC with the most physical resources (from line \ref{33} to line \ref{34} of algorithm \ref{alg3}). After mapping DCs to slices, If a slice does not admit to a specific DC, it remains for the next placement. In the next placement the remaining slices, map to more than one DC according to their requirements. (from line \ref{35} to line \ref{36} of algorithm \ref{alg3}). At the end, if DC with the lowest physical resources is free and can be served the slices of a DC with the highest physical resource, the slices remapped to new DC with the lowest physical resource since it has the lowest power consumption (line \ref{37} of algorithm \ref{alg3}).
\begin{algorithm}
\caption{Plecement of Physical resources into Virtual resources}\label{alg3}
\begin{algorithmic}[1]
\State Sort Slices according to $\hat{\Omega}_{s}^{tot} , \forall s$ in descending order.\label{31}
\State Sort DCs according to $\hat{\tau}_j , \forall j$ in descending order. \label{32}
\State $\boldsymbol{Y} = \boldsymbol{0}$
\For {$d \gets 1$ to $D_c$}\label{33}
\For {$s \gets 1$ to $S$}
\If {$\sum_{d=1}^{D_c}y_{s,d}==0$ and $\bar{\Omega}_{\mathfrak{z}(s)}^{tot} \leq \tau_{\mathfrak{z}_j} \forall \mathfrak{z}, \forall s$  }
\State Set $y_{s,d} = 1$;
\State {$\tau_j^{\mathfrak{z}}$ $\gets$ {$\tau_j^{\mathfrak{z}} - \bar{\Omega}_{\mathfrak{z},s}^{tot}$}} $\mathfrak{z}\in \{M,S,C\}$
\EndIf
\State \textbf{end if}
\EndFor
\State \textbf{end for}
\EndFor
\State \textbf{end for} \label{34}
\State  $ind_{rem} = \{s|({\sum_{d=1}^{D_c}y_{s,d}==0})\}$ \label{35}
\State Sort remaining amount of DCs same as before in descending order.
\State Sort remaining slices same as before in descending order.
\For {$r \gets 1$ to $S_{rem}$}
\For {$n \gets 1$ to $D_c$}
\State Set $y_{s,d} = 1$;
\State {$\bar{\Omega}_{{\mathfrak{z}},s}^{tot}$ $\gets$ {$\bar{\Omega}_{{\mathfrak{z}},s}^{tot}- \tau_j^{\mathfrak{z}}$}}
\If {$\bar{\Omega}_{s}^{tot}==0$}
\State Set $y_{s,d} = 1$;
\State {$\tau_j^{\mathfrak{z}}$ $\gets$ {$\tau_j^{\mathfrak{z}} - \bar{\Omega}_{\mathfrak{z},s}^{tot}$}} $\mathfrak{z}\in \{M,S,C\}$
\State {break inner loop}
\EndIf
\State \textbf{end if}
\EndFor
\State \textbf{end for}
\EndFor
\State \textbf{end for} \label{36}
\State Remapping DCs must be done to prevent wasting Energy \label{37}
\end{algorithmic}
\end{algorithm}
\section{Numerical Results}\label{simul}
In this section, numerical results for the main problem are depicted.
 \begin{table}
 \caption {Simulation Parameter} \label{table:1a}
 \begin{center}
  \begin{tabular}{||c c ||}
  \hline
Parameter & Value \\ [0.5ex]
  \hline\hline
  Noise power & -174dBm\\
  \hline
  Bandwidth & 120 KHZ \\
  \hline
 Maximum transmit Power of each RU & 40dBm \\
  \hline
  Minimum delay &  300usec \\
  \hline
  Maximum fronthaul capacity  & 200 bits/sec/Hz \\
   \hline
  Minimum data rate &  10 bits/sec/Hz \\ [1ex]
  \hline
 \end{tabular}
 \end{center}
 \end{table}
In Fig. \ref{fig:f1a}, energy efficiency is depicted for two different numbers of services with different mean number of UEs in each service The parameters of simulations are listed in Table \ref{table:1a}. The optimal method is obtained by using MOSEK toolbox to obtain $\boldsymbol{A}$ and CVX toolbox to obtain $\boldsymbol{P}$. They iteratively update $\boldsymbol{A}$, $\eta$ and $\boldsymbol{P}$. The optimal method is $0.1$ bits/J/Hz better than the proposed method for $V = 6$ and $E[U_v] = 10$, and also, $0.09$ bits/J/Hz better than the proposed method for $V = 3$ and, $E[U_v] = 10$. As it is shown, the energy efficiency is increased as mean number of UEs rises.
\begin{figure}
  \centering
    \includegraphics[width=\linewidth]{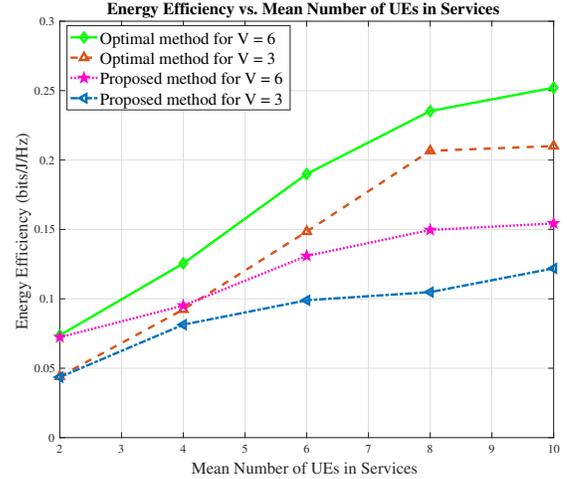}
  \caption{Energy Efficiency vs. Mean Number of UEs in each Service}
  \label{fig:f1a}
\end{figure}
In Fig. \ref{fig:f1}, the ratio of admitted slices is demonstrated for two different numbers of DCs with the different number of slices. The parameters of the simulation are listed in Table \ref{table:1}. We also set $w_C = 320$, $w_S = 100$, $w_M =1$. In this simulation, the second term of \eqref{eqpsi}, is important and the designed parameter $\nu$ is supposed to be high.
Also, it is assumed that just one DC can serve each slice, and each slice is not admitted by more than one DC. 
The proposed method is based on Algorithm \ref{alg3}, and the optimal method is obtained using MOSEK toolbox.
\begin{small}
 \begin{table}
 \caption {Simulation Parameter} \label{table:1}
 \begin{center}
  \begin{tabular}{||c c ||}
  \hline
  Parameter & Value \\ [0.5ex]
  \hline\hline
  Mean of CPU for DCs & 320GHz\\
  \hline
  Mean of Memory for DCs & 1T\\
  \hline
 Mean of Storage for DCs & 100T \\
  \hline
   Mean of CPU for Slices & 32GHz\\
  \hline
  Mean of Memory for Slices & 100G\\
  \hline
 Mean of Storage for Slices & 10T \\ [1ex]
  \hline
 \end{tabular}
 \end{center}
 \end{table}
\end{small}
When we have two DCs, the proposed method and optimal method have approximately the same ratio of admitted slices. But by increasing the number of DCs to five, the performance of the proposed method reduced. Using five DCs, the difference between the
proposed method and the optimal method in the worst case (44 slices) is about $23$ percentage.
\begin{figure}
  \centering
    \includegraphics[width=\linewidth]{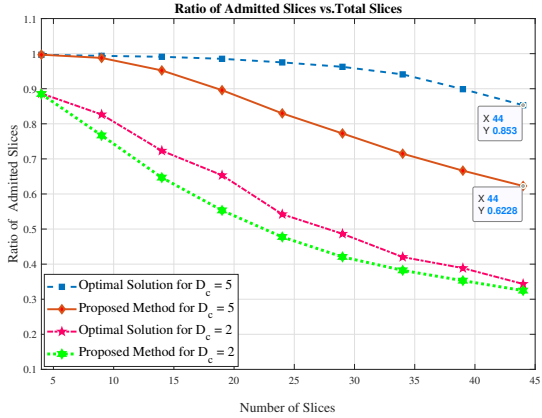}
  \caption{Ratio of Admitted Slices connected to just one DC vs. Total slices}
  \label{fig:f1}
\end{figure}
In Fig. \ref{fig:f2}, the normalized resource consumption is depicted versus the number of slices. The parameters for simulation are listed in Table \ref{table:1}. We also set $w_C = 320$, $w_S = 100$, $w_M =1$). In this simulation, suppose the number of DCs is entirely enough to cover all slices and the designed parameter $\nu$ is assumed to be low, so we focused on the first term of \eqref{eqpsi}. However, if any slice remains, it can be served by more than one DCs. The optimality of the placement of DC resources to slices is measured based on the power consumption of DCs. It is shown the amount of resources of DCs which are not used. For ten slices, the difference between the optimal solution and the proposed solution is about $15$ percent.
\begin{figure}
  \centering
    \includegraphics[width=\linewidth]{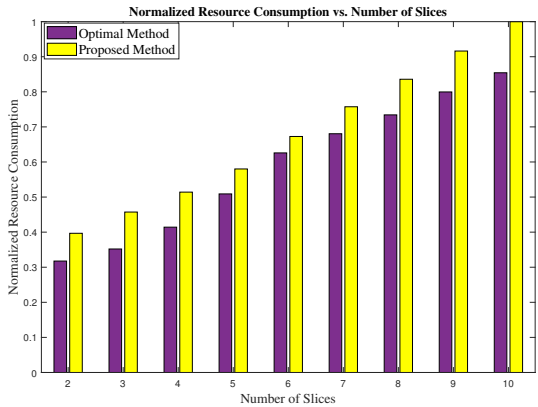}
  \caption{Normalized Resource Consumption vs. Number of Slices}
  \label{fig:f2}
\end{figure}
\section{Conclusion}
In this paper, joint network slicing and power allocation were considered in an ORAN system. It is assumed that UEs are classified based on their requirements into services. Also, there is a number of slices serving the services. Each slice consists of PRBs, RUs, and VNFs that support CU and DU. The limited fronthaul capacity is considered for the fiber links between RUs and DU.
The target was to maximize the sum-rate and minimize the power consumption and energy cost of data centers simultaneously.
The problem is decomposed into two sub-problems. Each sub-problem is solved separately by a heuristic algorithm. Using numerical results, we validate the heuristic method and study the performance of the algorithms. 

\bibliographystyle{IEEEtran}
\bibliography{ref}
\end{document}